\documentclass[aps,pra,preprint,nofootinbib,showpacs,superscriptaddress]{revtex4-1}
\usepackage{amsmath}
\usepackage{amsthm}
\usepackage{amssymb}
\usepackage{subfigure}
\usepackage{graphicx}
\usepackage{enumitem}
\usepackage[qm]{qcircuit}
\usepackage{amsfonts}
\usepackage{braket}
\usepackage{hyperref}
\usepackage{xcolor}
	\hypersetup
	{
		colorlinks,%
		citecolor=green,%
		linkcolor=blue,%
		urlcolor=blue,%
	}

	\newtheorem{theorem}{Theorem}

\newcommand{\gv}[1]{\ensuremath{\text{\boldmath$ #1 $}}}
\newcommand{\abs}[1]{\left| #1 \right|} 
\let\baraccent=\= 
\renewcommand{\=}[1]{\stackrel{#1}{=}} 

\newcommand{\aref}[1]{\hyperref[#1]{Appendix~\ref{#1}}}

\begin{document}

\title{Quantum Fourier Transform in Computational Basis}
\author{S. S. Zhou }
\affiliation{{School of Physics, The University of Western Australia, Crawley WA 6009, Australia}}
\affiliation{{Kuang Yaming Honors School, Nanjing University, Nanjing, 210093, China}}
\email{jingbo.wang@uwa.edu.au}
\author{T. Loke }
\author{J. A. Izaac }
\author{J. B. Wang }
\affiliation{{School of Physics, The University of Western Australia, Crawley WA 6009, Australia}}
\keywords{}

\begin{abstract}
The conventional Quantum Fourier Transform, with exponential speedup compared to the classical Fast Fourier Transform,
has played an important role in quantum computation as a vital part of many quantum algorithms (most prominently, the Shor's factoring algorithm).
However, situations arise where it is not sufficient to encode the Fourier coefficients within the quantum amplitudes,
for example in the implementation of control operations that depend on Fourier coefficients.
In this paper, we detail a new quantum algorithm to encode the Fourier coefficients in the computational basis,
with fidelity $1 - \delta$ and desired precision $\epsilon$. Its time complexity depends polynomially on $\log(N)$, where $N$ is the problem size, and linearly on $1/\delta$ and $1/\epsilon$.
We also discuss an application of potential practical importance, namely the  simulation of circulant Hamiltonians.
\end{abstract}


\maketitle

\section{Introduction}
\label{sec:introduction}

Since the milestone introduction of Shor's quantum factoring algorithm \cite{shor1997} allowing prime number factorization
with complexity $\mathcal{O}(\text{polylog}N)$ -- an exponential speedup compared to the fastest-known classical algorithms
 -- there has been an increasing number of quantum algorithm discoveries harnessing the unique properties of quantum mechanics
 in order to achieve significant increases in computational efficiency. The use of the Quantum Fourier Transform (QFT) \cite{deutsch1985}
 in Shor's factoring algorithm is integral to the resulting speedup.

The Fast Fourier Transform (FFT), an efficient classical implementation of the discrete Fourier transform (DFT), is a hugely
important algorithm, with classical uses including signal processing and frequency analysis \cite{bergland1969}. Due to its widespread ubiquity
and efficiency (with scaling $\mathcal{O}(N \log N)$), it has been regarded to be the one of the most
important non-trivial classical algorithms  \cite{cleve2000}.

The QFT (with complexity $\mathcal{O}((\log N)^2)$) algorithm is the natural extension of the DFT to the quantum regime,
with exponential speedup realized compared to the FFT ($\mathcal{O}(N\log N)$), due to superposition and quantum parallelism.
The QFT is essentially identical to the FFT
in that it performs a DFT on a list of complex numbers, but the result of the QFT is stored as amplitudes of a quantum state
vector. In order to extract the individual Fourier components, measurements need to be performed on the quantum state vector.
As such, the QFT is not directly useful for determining the Fourier transformed amplitudes of the original list of numbers.
{However, the QFT is widely used as a subroutine in larger algorithms, including but not limited to Shor's algorithm \cite{shor1997},
quantum amplitude estimation \cite{kitaev1995} and quantum counting \cite{brassard1998,brassard2002}}.

Typically, there are two methods of encoding the result of a quantum algorithm: encoding within the computational basis
of the quantum state \cite{kitaev1995}, or encoding within the amplitudes of the quantum state \cite{deutsch1985}.
The QFT fits the latter criteria, and has been enormously successful, used as a foundation for a plethora of other
quantum algorithms -- for example in the fields of quantum chemistry and simulations \cite{benenti2008,kassal2008,szkopek2005},
signal and image processing \cite{hales2000,schuetzhold2003}, cryptography \cite{vandam2006} and
computer science \cite{cleve2000,jordan2005}.
However, situations arise where we need the Fourier coefficients in the computational basis, for example in order to
efficiently implement circulant Hamiltonians with quantum circuits~\cite{qiang2015}.

In this paper, we introduce a new quantum scheme for computing the Fourier transform and storing the results in the computational
basis, namely Quantum Fourier Transform in the Computational Basis (QFTC). We begin in \autoref{sec:defs} by defining the notations
and chosen conventions, before subsequently detailing the QFTC algorithm for computing the DFT in the computational basis in \autoref{sec:qftc}.
This section also includes a thorough analytic derivation of the complexity and error analysis.
One possible application of this algorithm, the implementation of circulant Hamiltonians,
is then discussed in \autoref{sec:applications}.
Finally, we present our conclusions in \autoref{sec:conclusion}.
In addition, we have provided supplementary material in the appendices,
detailing the quantum arithmetic necessary for the QFTC algorithm \aref{sec:arithmetic}
and the implementation of circulant matrix operators \aref{sec:circulantoperator}.

\section{Definitions and notations}
\label{sec:defs}

The DFT, applied to a unit vector $\gv{x}=(x_0~x_1~\cdots~x_{N-1})\in\mathbb{C}^N$, outputs a unit vector $\gv{y}=(y_0~y_1~\cdots~y_{N-1})$, where
\begin{equation}
y_k = \frac{1}{\sqrt{N}} \sum_{j=0}^{N-1} e^{2\pi i jk / N} x_j, ~~~k=0,1,\ldots,N-1.
\end{equation}
The QFT performs the Discrete Fourier Transform in amplitudes:
\begin{equation}
\sum_{j=0}^{N-1} x_j\ket{j} 
\rightarrow
\sum_{k=0}^{N-1} y_k\ket{k}.
\end{equation}
The QFTC, on the other hand, enables the Fourier transformed coefficients to be encoded in the computational basis:
\begin{equation}
\ket{0} 
\rightarrow
\frac{1}{\sqrt{N}} \sum_{k=0}^{N-1} \ket{k}\ket{y_k}
\end{equation}
where $\ket{y_k}$ corresponds to the fixed-point binary representation of number $y_k \in (-1,1)$ (see the complemental encoding in \aref{sec:arithmetic}).
In the QFTC algorithm, the value of $\gv{x}$ is provided by an oracle $O_x$
\begin{equation}
\label{eq:oracle}
O_x \ket{0} = \sum_j x_j \ket{j}.
\end{equation}
The oracle can be efficiently implemented if $\gv{x}$ is efficiently computable, or using the qRAM
which takes complexity $\log N$ under certain conditions \cite{giovannetti2008, grover2002,lloyd2013,kaye2004,soklakov2006}.
In fact, in the QFTC algorithm,
the controlled-$O_x$ gate ($\ket{0}\bra{0}\otimes{\mathbb{I}} + \ket{1}\bra{1}\otimes O_x$) is required.
(Otherwise, there will always be an indefinite total phase $e^{i\phi}$ of $\gv{x}$).
The number of calls to controlled-$O_x$ and its inverse will be included in the overall complexity of the QFTC algorithm.

Without loss of generality, we will assume the $y_k$ coefficients are real in the following sections. If this is not the case, we can always redefine the inputs as the following:
\begin{equation}
x'_j = \frac{x_j + x^*_{N-j}}{2} \quad(x_N = x_0 \text{~and the }x^*_j\text{ is the complex conjugate of }x_j.)
\end{equation}
for all $j$. Applying the DFT to $\gv{x'}$ then produces a purely real result, $y'_k = \text{Re} (y_k)$. The imaginary components $\text{Im}(y_k)$ can be derived analogously in the same fashion. In the following sections, we assume that $N = 2^L$, where $L$ is some integer, as in the conventional FFT and QFT algorithms.


\section{Quantum Fourier Transform in the Computational Basis}
\label{sec:qftc}

The steps involved in the QFTC algorithm are detailed below (with \autoref{fig:circuit1to4} depicting the circuit for \ref{item:prep_k}--\ref{item:phase}
and \autoref{fig:whole} for \ref{item:swap}--\ref{item:minus}).

\newcounter{saveenum}

\begin{enumerate}[label=\emph{Step~\arabic*},start=0]
\item Intialise all qubits, including ancillas, to $\ket{0}$.
\item Prepare the first register of $L$ qubits into an equal superposition of its computational basis states using a Hadamard transform:
\label{item:prep_k}
\begin{equation}
\ket{0^L} \xrightarrow{H^{\otimes L}} \frac{1}{\sqrt{N}}\sum_{k=0}^{N-1}\ket{k},
\end{equation}
where $k$ is represented in binary as $k_1k_2\cdots k_L$ with $L$ qubits.

\item Prepare an ancillary qubit in the third register as:
\label{item:prep_ancilla}
\begin{equation}
\ket{0} \xrightarrow{H}
\frac{1}{\sqrt{2}}(\ket{0}+\ket{1}).
\end{equation}

\item
Apply the controlled-$O_x$ and controlled-Hadamard gates to the second register of $L$ qubits, conditional on the ancillary qubit state:
\begin{equation}
\label{item:prep_j}
\ket{0^L}\frac{1}{\sqrt{2}}\big(\ket{1}+\ket{0}\big)
\xrightarrow{ O_x \otimes \ket{1}\bra{1} + H^{\otimes L} \otimes \ket{0}\bra{0}}
\frac{1}{\sqrt{2}}\big(\sum_{j=0}^{N-1}x_j\ket{j}\ket{1}+\frac{1}{\sqrt{N}}\sum_{j=0}^{N-1}\ket{j}\ket{0}\big),
\end{equation}
where $j$ is represented in binary as $j_1j_2\cdots j_L$ with $L$ digits.

\item
\label{item:phase}
Apply a controlled phase operator on these three registers (with details given in \autoref{fig:phase}):
\begin{multline}
\label{eq:def-phik}
\frac{1}{\sqrt{N}}\sum_{k,j=0}^{N-1}\ket{k}\frac{1}{\sqrt{2}}\big(x_j\ket{j}\ket{1}+\frac{1}{\sqrt{N}}\ket{j}\ket{0}\big)
\xrightarrow{\sum_{j,k} e^{2\pi ijk/N} \ket{k}\bra{k}\otimes\ket{j}\bra{j}\otimes\ket{1}\bra{1}} \\
\frac{1}{\sqrt{N}}\sum_{k,j=0}^{N-1}\ket{k}\frac{1}{\sqrt{2}}\big(x_je^{2\pi ijk/N}\ket{j}\ket{1} + \frac{1}{\sqrt{N}}\ket{j}\ket{0} \big)
\equiv \frac{1}{\sqrt{N}}\sum_{k=0}^{N-1}\ket{k}\ket{\phi_k}.
\end{multline}

\setcounter{saveenum}{\value{enumi}}

\end{enumerate}

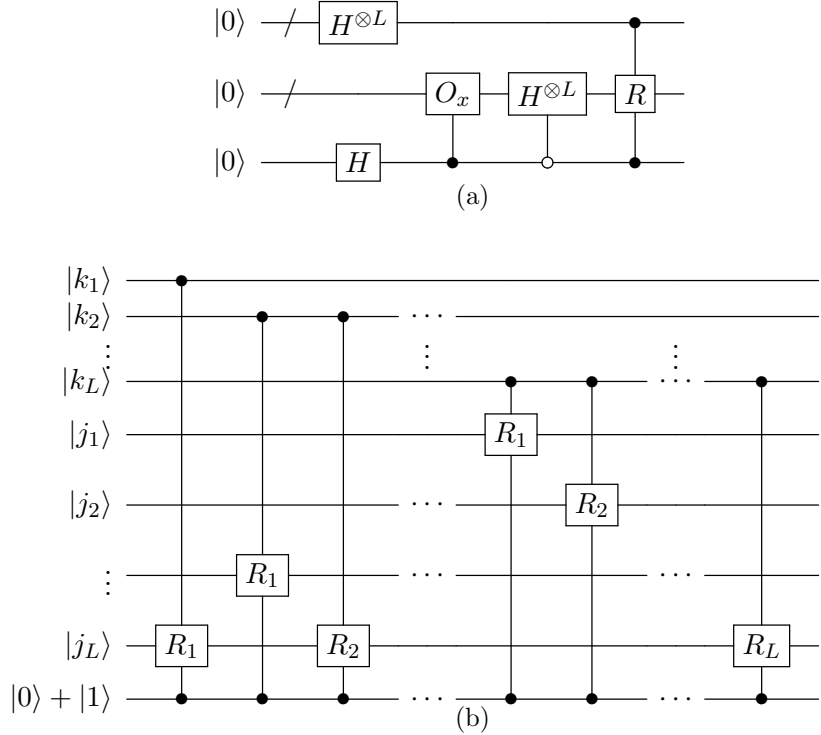
\begin{figure}[ht]
\subfigure[\label{fig:prep}]{$
\Qcircuit @C=1em @R=1em
{
\lstick{\ket{0}}	&	{/}\qw	&	\gate{H^{\otimes L}}	&	\qw	&	\qw	&	\ctrl{1}	&	\qw	\\
\lstick{\ket{0}}	&	{/}\qw	&	\qw	&	\gate{O_x}	&	\gate{H^{\otimes L}}	&	\gate{R}	&	\qw	\\
\lstick{\ket{0}}	&	\qw	&	\gate{H}	&	\ctrl{-1}	&	\ctrlo{-1}	&	\ctrl{-1}	&	\qw	
}
$}\\
\vspace{0.2in}
\subfigure[\label{fig:phase}]{$
\Qcircuit @C=1em @R=1em
{
\lstick{\ket{k_1}}	&	\ctrl{7}	&	\qw	&	\qw	&	\qw	&	\qw	&	\qw	&	\qw	&	\qw	&	\qw	&	\qw	&	\qw	&	\qw	&	\qw	\\
\lstick{\ket{k_2}}	&	\qw	&	\ctrl{5}	&	\ctrl{6}	&	\qw	&	\cdots	&		&	\qw	&	\qw	&	\qw	&	\qw	&	\qw	&	\qw	&	\qw	\\
\lstick{\vdots}	&		&		&		&		&	\vdots	&		&		&		&		&	\vdots	&		&		&		\\
\lstick{\ket{k_L}}	&	\qw	&	\qw	&	\qw	&	\qw	&	\qw	&	\qw	&	\ctrl{1}	&	\ctrl{2}	&	\qw	&	\cdots	&		&	\ctrl{4}	&	\qw	\\
\lstick{\ket{j_1}}	&	\qw	&	\qw	&	\qw	&	\qw	&	\qw	&	\qw	&	\gate{R_1}	&	\qw	&	\qw	&	\qw	&	\qw	&	\qw	&	\qw	\\
\lstick{\ket{j_2}}	&	\qw	&	\qw	&	\qw	&	\qw	&	\cdots	&		&	\qw	&	\gate{R_2}	&	\qw	&	\qw	&	\qw	&	\qw	&	\qw	\\
\lstick{\vdots}	&	\qw	&	\gate{R_1}	&	\qw	&	\qw	&	\cdots	&		&	\qw	&	\qw	&	\qw	&	\cdots	&		&	\qw	&	\qw	\\
\lstick{\ket{j_L}}	&	\gate{R_1}	&	\qw	&	\gate{R_2}	&	\qw	&	\qw	&	\qw	&	\qw	&	\qw	&	\qw	&	\qw	&	\qw	&	\gate{R_L}	&	\qw	\\
\lstick{\ket{0}+\ket{1}}	&	\ctrl{-1}	&	\ctrl{-2}	&	\ctrl{-1}	&	\qw	&	\cdots	&		&	\ctrl{-4}	&	\ctrl{-3}	&	\qw	&	\cdots	&		&	\ctrl{-1}	&	\qw	\\
}
$}
\caption{(a) the quantum circuit for \ref{item:prep_k}--\ref{item:phase}; (b) Detailed quantum gates to implement the controlled phase operator in \ref{item:phase}.
Here $R_\ell = \ket{0}\bra{0}+e^{2\pi i/2^{\ell}}\ket{1}\bra{1}$.}
\label{fig:circuit1to4}
\end{figure}

Using the Hadamard gate and the pauli-Z gate, we can prepare two additional registers in the quantum states $\ket{\phi^\pm}$:
\begin{align}
\ket{0^{L+1}}\xrightarrow{(+):~H^{\otimes L}\otimes{H};~(-):~H^{\otimes L}\otimes{ZH}}\ket{\phi^\pm} =
\frac{1}{\sqrt{2}}\big(\sum_{j=0}^{N-1}\frac{\pm 1}{\sqrt{N}}\ket{j}\ket{1} + \sum_{j=0}^{N-1}\frac{1}{\sqrt{N}}\ket{j}\ket{0} \big).
\end{align}
Note that $\ket{\phi_k} = \frac{1}{\sqrt{2}}\big(\sum_{j=0}^{N-1}x_j^{2\pi i j k/N}\ket{j}\ket{1} + \sum_{j=0}^{N-1}\frac{1}{\sqrt{N}}\ket{j}\ket{0} \big)
$ in \autoref{eq:def-phik},
\begin{equation}
\abs{\braket{\phi^\pm|\phi_{k}}}^2 =
\frac{1}{4} (y_k^2 + 1) \pm \frac{y_k}{2},
\end{equation}
and
\begin{equation}
\abs{\braket{\phi^+|\phi_{k}}}^2 - \abs{\braket{\phi^-|\phi_{k}}}^2 = y_k,
\end{equation}
which leads to the following steps (as detailed in \autoref{fig:whole}).

\begin{figure}[ht]
\[
\Qcircuit @C=1em @R=1em{
	&		&		&		&		&		&	\lstick{\ket{0}}	&	\multigate{1}{2\sin^2(\pi \cdot )-1}	&	\multigate{8}{\mathrm{\Sigma^-}}	&	\rstick{\ket{\abs{\braket{\phi^+|\phi_k}}^2}}\qw	\\
	&	\lstick{\ket{0}}	&	{/}\qw	&	\qw	&	\gate{H^{\otimes p_0}}	&	\ctrl{1}^{~~\ket{\ell}}	&	\gate{\text{QFT}^\dagger}	&	\ghost{2\sin^2(\pi \cdot )-1}\qw	&		&		\\
\lstick{\ket{0}_1}	&	\gate{H}	&	\ctrl{1}	&	\gate{H}	&	\qw	&	\multigate{2}{(Q^+_k)^\ell}	&	\qw	&		&		&		\\
\lstick{\ket{\phi_k}_2}	&	{/}\qw	&	\multigate{1}{\mathrm{SWAP}}	&	\qw	&	\qw	&	\ghost{(Q^+_k)^\ell}	&	\qw	&		&		&		\\
\lstick{\ket{\phi^+}_3}	&	{/}\qw	&	\ghost{\mathrm{SWAP}}	&	\qw	&	\qw	&	\ghost{(Q^+_k)^\ell}	&	\qw	&		&		&		\\
	&		&		&		&		&		&	\lstick{\ket{0}}	&	\multigate{1}{2\sin^2(\pi \cdot )-1}	&	\ghost{\mathrm{\Sigma^-}}	&	\rstick{\ket{\abs{\braket{\phi^-|\phi_k}}^2}}\qw	\\
	&	\lstick{\ket{0}}	&	{/}\qw	&	\qw	&	\gate{H^{\otimes p_0}}	&	\ctrl{1}^{~~\ket{\ell}}	&	\gate{\text{QFT}^\dagger}	&	\ghost{2\sin^2(\pi \cdot )-1}\qw	&		&		\\
\lstick{\ket{0}}	&	\gate{H}	&	\ctrl{1}	&	\gate{H}	&	\qw	&	\multigate{2}{(Q^-_k)^\ell}	&	\qw	&		&		&		\\
\lstick{\ket{\phi_k}}	&	{/}\qw	&	\multigate{1}{\mathrm{SWAP}}	&	\qw	&	\qw	&	\ghost{(Q^-_k)^\ell}	&	\qw	&	\lstick{\ket{0}}	&	\ghost{\mathrm{\Sigma^-}}	&	\rstick{\ket{y_k}}\qw	\\
\lstick{\ket{\phi^-}}	&	{/}\qw	&	\ghost{\mathrm{SWAP}}	&	\qw	&	\qw	&	\ghost{(Q^-_k)^\ell}	&	\qw\gategroup{2}{5}{10}{7}{1.2em}{--}	&		&		&		\\
	&		&		&		&		&	~~~~~\mbox{amplitude estimation}	&		&		&		&		\\	
}
\]
\caption{the quantum circuit for \ref{item:swap}--\ref{item:minus}.
The $\boxed{\Sigma^-}$ gate transforms $\ket{\alpha}\ket{\beta}\ket{0}$ into $\ket{\alpha}\ket{\beta}\ket{\alpha-\beta}$ (see \aref{sec:arithmetic}).
}
\label{fig:whole}
\end{figure}
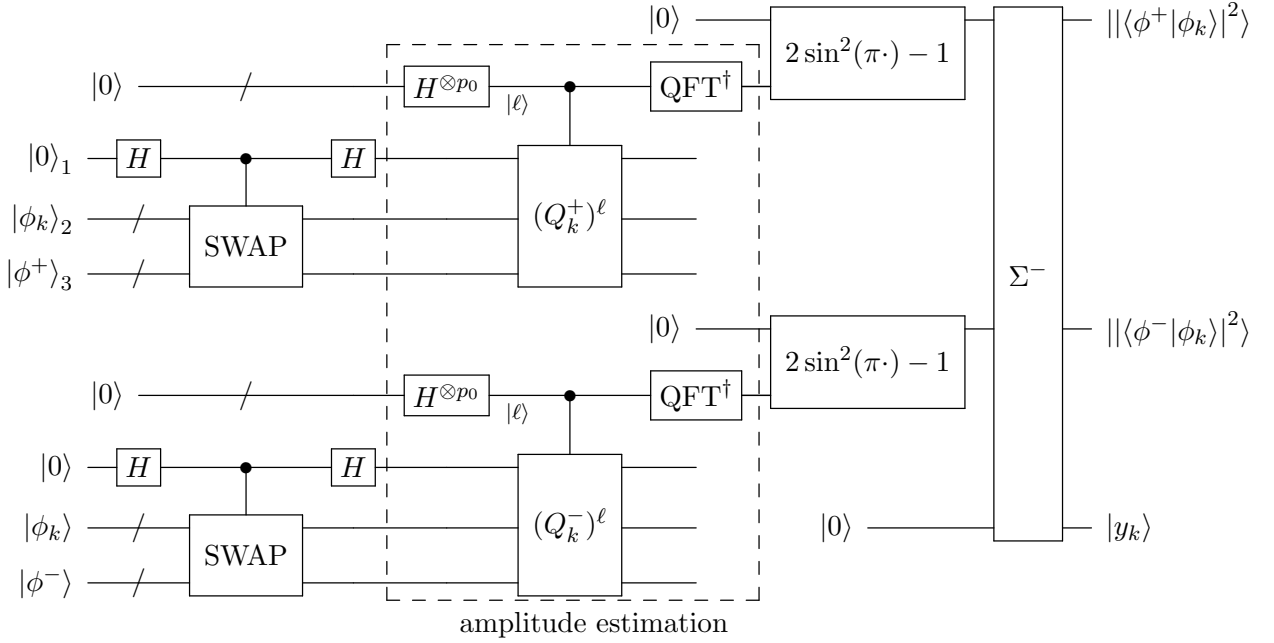

\begin{enumerate}[label=\emph{Step~\arabic*}]

\setcounter{enumi}{\value{saveenum}}

\item
\label{item:swap}
Prepare $\ket{\phi^+}$ and perform the swap test \cite{buhrman2001} on $\ket{\phi_{k}}$ and $\ket{\phi^+}$. We get
\begin{equation}
\label{eq:innerproduct}
\frac{1}{2}\ket{0}\big(\ket{\phi_k}\ket{\phi^+}+\ket{\phi^+}\ket{\phi_k}\big)
+ \frac{1}{2}\ket{1}\big(\ket{\phi_k}\ket{\phi^+}-\ket{\phi^+}\ket{\phi_k}\big) \equiv \ket{\psi^+_k}
\end{equation}
for all values of $k$.

\item  Run {amplitude estimation} of $Q^+_k$, for all $k$, on state $\ket{\psi^+_k}$ as defined below:
\label{item:phase_estimation1}
\begin{equation}
\ket{\psi^+_k}
\rightarrow
\Ket{\frac{\theta_k}{\pi} }\Ket{\psi^{\uparrow}_k}  +
\Ket{1 - \frac{\theta_k}{\pi}}\Ket{\psi^{\downarrow}_k}.
\end{equation}

\item  Compute $\abs{\braket{\phi^+|\phi_{k}}}^2 = (y_k^2 + 1)/4 + y_k/2$ using the quantum multiply-adder and sine gate (see \aref{sec:arithmetic} for details),
for all values of $k$:
\label{item:phase_estimation2}
\begin{equation}
\Ket{\frac{\theta_k}{\pi}} \Ket{\psi^{\uparrow}_k}  +
\Ket{1 - \frac{\theta_k}{\pi}} \Ket{\psi^{\downarrow}_k}
\rightarrow
\Ket{\abs{\braket{\phi^+|\phi_{k}}}^2} \ket{\psi^+_k} .
\end{equation}

\setcounter{saveenum}{\value{enumi}}

\end{enumerate}

In the above description of \ref{item:swap}--\ref{item:minus},
\begin{equation}
\ket{\psi^+_k} = \sin{\theta_k} \ket{\psi^{0}_k} + \cos{\theta_k} \ket{\psi^{1}_k}
\end{equation}
where $\ket{\psi^{0}_k}$ corresponds to the part of $\ket{\psi_k^+}$ whose first qubit is $\ket{0}$, $\ket{\psi^{1}_k}$ corresponds to the part of $\ket{\psi_k^+}$
whose first qubit is $\ket{1}$.
It can be easily calculated from \autoref{eq:innerproduct} that $\sin^2\theta_k = (1 + \abs{\braket{\phi^+|\phi_{k}}}^2)/2 $ where $\theta_k \in [0,\pi/2]$.
We define $Q^+_k := -\mathcal{A}^+_k S_0 (\mathcal{A}^+_k)^{\dagger} S_\chi$,
where $\mathcal{A}^+_k$ is the unitary operator performing $\ket{0}_{123} \xrightarrow{\mathcal{A}^+_k} \ket{\psi^+_k}$, $S_0 = \mathbb{I} - 2\ket{0}_{123}\bra{0}_{123}$ and $S_\chi = \mathbb{I} - 2\ket{0}_{1}\bra{0}_{1}$
(subscripts denote labels of registers shown in \autoref{fig:whole}).
According to the amplitude estimation algorithm \cite{brassard2002},
\begin{equation}
(Q^+_k)^{\ell}\ket{\psi^+_k}
= \sin(2\ell+1)\theta_k \Ket{\psi^{0}_k} + \cos(2\ell+1)\theta_k \Ket{\psi^{1}_k}.
\end{equation}
For any $\ell \in \mathbb{N}$,
$Q^+_k$ acts as a rotation in $2$-dimensional space $\mathrm{Span}\{\Ket{\psi^{0}_k},\Ket{\psi^{1}_k}\}$,
and it has eigenvalues $e^{\pm i 2\theta_k}$ with eigenstates $\ket{\psi^{\uparrow,\downarrow}_k}$ (un-normalized).
Therefore we can generate the state
\begin{equation}
\ket{\psi^+_k} = \ket{\psi^{\uparrow}_k} + \ket{\psi^{\downarrow}_k} \xrightarrow{\text{phase estimation}} \Ket{\frac{\theta_k}{\pi}}\Ket{\psi^{\uparrow}_k}
+\Ket{1 - \frac{\theta_k}{\pi}} \Ket{\psi^{\downarrow}_k},
\end{equation}
by running amplitude estimation of $\mathcal{A}^+_k$ on $\ket{\psi^+_k}$, and get $\Ket{\abs{\braket{\phi^+|\phi_{k}}}^2} = \ket{2\sin^2\theta_k-1}$
using the quantum multiply-adder and sine gate (see \aref{sec:arithmetic}).
The quantum circuit of amplitude estimation procedure is marked out in \autoref{fig:whole}.

\begin{enumerate}[label=\emph{Step~\arabic*}]

\setcounter{enumi}{\value{saveenum}}

\item
\label{item:pre-minus}
Repeat \ref{item:prep_ancilla}--\ref{item:phase_estimation2} in other registers, with $\ket{\phi^+}$ and $\mathcal{A}^+_k$
replaced by $\ket{\phi^-}$ and $\mathcal{A}^-_k$, we get
\begin{equation}
\frac{1}{\sqrt{N}}\sum_{k=0}^{N-1}\ket{k}\Ket{\abs{\braket{\phi^+|\phi_{k}}}^2}\ket{\psi^+_k}
\rightarrow \frac{1}{\sqrt{N}}\sum_{k=0}^{N-1}\ket{k}\Ket{\abs{\braket{\phi^+|\phi_{k}}}^2}\ket{\psi^+_k}\Ket{\abs{\braket{\phi^-|\phi_{k}}}^2}\ket{\psi^-_k}.
\end{equation}

\item
\label{item:minus}
Calculate $\abs{\braket{\phi^+|\phi_{k}}}^2$ minus $\abs{\braket{\phi^-|\phi_{k}}}^2$ and encode the result in $p_0+1$ qubits (error $\epsilon = 2^{-p_0}$), using the quantum adder described in \aref{sec:arithmetic}:
\begin{multline}
\frac{1}{\sqrt{N}}\sum_{k=0}^{N-1}\ket{k}\Ket{\abs{\braket{\phi^+|\phi_{k}}}^2}\ket{\psi^+_k}\Ket{\abs{\braket{\phi^-|\phi_{k}}}^2}\ket{\psi^-_k}\ket{0^{p_0+1}}
\rightarrow\\
\frac{1}{\sqrt{N}}\sum_{k=0}^{N-1}\ket{k}\Ket{\abs{\braket{\phi^+|\phi_{k}}}^2}\ket{\psi^+_k}\Ket{\abs{\braket{\phi^-|\phi_{k}}}^2}\ket{\psi^-_k}\ket{y_k}
\equiv
\frac{1}{\sqrt{N}} \sum_{k=0}^{N-1} \ket{k}\ket{\Psi^{\mathrm{ancilla}}_k}\ket{y_k}.
\end{multline}

\item
\label{item:uncompute}
Uncompute the ancillas using the inverse algorithm of \ref{item:prep_ancilla}--\ref{item:pre-minus}:
\begin{equation}
\frac{1}{\sqrt{N}} \sum_{k=0}^{N-1} \ket{k}\ket{\Psi^{\mathrm{ancilla}}_k}\ket{y_k}
\rightarrow
\frac{1}{\sqrt{N}} \sum_{k=0}^{N-1} \ket{k}\ket{y_k}.
\end{equation}
\setcounter{saveenum}{\value{enumi}}

\end{enumerate}

\section{Complexity analysis}
\label{sec:Complexity}

\begin{theorem}[QFTC]
\label{thm:qftc}
The required quantum state $\frac{1}{\sqrt{N}} \sum_{k=0}^{N-1} \ket{k}\ket{y_k}$
can be prepared
to accuracy $\epsilon$
\footnote{$\abs{y_k - \tilde{y}_k} < \epsilon$, where $\tilde{y}_k$ is the truncated value of $y_k$ with accuracy $2^{-p_0}$.}
with fidelity $1 - \delta$
\footnote{$\abs{\bra{\Psi^{\text{final}}}(\frac{1}{\sqrt{N}} \sum_{k=0}^{N-1} \ket{k}\ket{\tilde{y}_k})} \geq 1 - \delta$, where $\ket{\Psi^{\text{final}}}$ is the state obtained through the QFTC algorithm.}
using
$\mathcal{O}\big((\log N)^2/(\delta\epsilon)\big)$
one- or two-qubit gates, and $\mathcal{O}\big(1/(\delta\epsilon))$ calls of controlled-$O_x$ and its inverse.
\end{theorem}

\begin{proof}

First, we consider the complexity involved in $\mathcal{A}^+_k$ (described in \ref{item:prep_ancilla}--\ref{item:swap}).
It contains Hadamard gates, controlled phase operators and swap gates which can be constructed using $\mathcal{O}\big((\log N)^2\big)$
one- or two-qubit gates and only one call of controlled-$O_x$.

The subsequent amplitude estimation block needs $\mathcal{O}(1/(\delta\varepsilon))$
applications of $Q^+_k = -\mathcal{A}^+_kS_0(\mathcal{A}^+_k)^{\dagger}S_\chi$
to get accuracy $\varepsilon$ with fidelity at least $1-\delta$ \cite{brassard2002,nielsen2010quantum}.
We then use the quantum multiply-adder and sine gate to get the value of $\abs{\braket{\phi^+|\phi_{k}}}^2 = \frac{1}{4}(1+y_k^2)+y_k/2$
for different $\ket{k}$'s in the computational basis. Using the similar procedure to obtain $\abs{\braket{\phi^-|\phi_{k}}}^2$,
we get $y_k = \abs{\braket{\phi^+|\phi_{k}}}^2 - \abs{\braket{\phi^-|\phi_{k}}}^2$ finally.
Since the derivative of $\sin x$ is always smaller than one, we set $\varepsilon = \Theta(\epsilon)$
in order to guarantee accuracy $\epsilon$ in $y_k$.
As detailed in \aref{sec:arithmetic}, the quantum multiply-adders and sine gates have complexity $\mathcal{O}(\text{polylog}(1/\varepsilon))$
which is smaller than $\mathcal{O}(1/\epsilon)$ in amplitude estimation.
Therefore, the complexity of these gates can be omitted.

The total complexity of the proposed circuit will be
$\mathcal{O}\big((\log N)^2/(\delta\epsilon) \big)$ one- or two-qubit gates, and $\mathcal{O}\big(1/(\delta\epsilon) \big)$ calls of controlled-$O_x$ and its inverse.
\end{proof}

Note that we do not calculate $y_k$ directly from $\abs{\braket{\phi^+|\phi_{k}}}^2$ to avoid having to implement a square-root gate, which must be designed carefully
due to the infinite derivative of the square root function at zero. Also we do not use gates to transform the value of $y_k$ directly into the amplitude in quantum state
like many other algorithms do \cite{grover1998,abrams1999,brassard2011,lloyd2013,aimeur2013,wiebe2014}. Instead, we put the value
$\frac{1}{4}(1+y_k^2)+y_k/2$ into the amplitude in order to take the sign of $y_k$ into account.
Throughout the proposed QFTC algorithm, $\ket{k}$ in the first register is used to control the application of quantum operators acting on other registers,
giving us the advantage of parallel calculating $y_k$ for all $k$,
which is the main reason the QFTC algorithm outshines the FFT in complexity.
Though values of $y_k$'s cannot be obtained by a single measurement of $\frac{1}{\sqrt{N}}\sum_{k=0}^{N-1}\ket{k}\ket{y_k}$,
they can be used widely in subsequent quantum computation once they are encoded in the computational basis.

\section{Application}
\label{sec:applications}

One important family of operators are the circulant matrices
which have found important applications in, for example, photonic quantum walks \cite{Delanty2012}, investigation on quantum supremacy \cite{qiang2015}, biochemical modelling \cite{yoneda2014}, vibration analysis~\cite{Olson2014}, and parallel diagnostic algorithm for super-computing \cite{Cheng2013}.

Circulant matrices are defined as follows \cite{golub2012}:
\begin{equation}
\label{eq:circulant}
C =
\begin{pmatrix}
c_0 &c_1 &\cdots & c_{N-1} \\
c_{N-1} &c_0 &\cdots &c_{N-2}\\
\vdots & \vdots & \ddots & \vdots\\
c_1 & c_2 & \cdots & c_0
\end{pmatrix},
\end{equation}
using an $N$-dimensional vector $\gv{c} = (c_0~c_1~\cdots~c_{N-1})$.
Such matrices are diagonalizable by the discrete Fourier transform (DFT), i.e.
\begin{equation}
C=
F \Lambda F^\dagger,
\end{equation}
where $F$ is the Fourier matrix with $F_{kj} = e^{2\pi ijk/N}/\sqrt{N}$, and $\Lambda$ is a diagonal matrix of eigenvalues given by $\Lambda_k = \sqrt{N} \big( F (c_0~c_1~\cdots~c_{N-1})^{\dagger} \big)_k \equiv \sqrt{N}F_k$. (Note that the condition that $C$ is Hermitian (in order to be a Hamiltonian) is equivalent to our assumption in \autoref{sec:defs} that the Fourier coefficients $F_k$ are real.) Due to this property, we are able to implement circulant quantum operators (non-unitary in general) using the conventional QFT through the manipulation of amplitudes, as detailed in \aref{sec:circulantoperator}.

However, this approach cannot be used directly for simulation of (non-sparse) circulant Hamiltonians, where we need to perform
\begin{equation}
e^{-iCt} \sum_{k=0}^{N-1}s_k\ket{k} = e^{-iCt} \ket{s} = \text{QFT} \;e^{-i\Lambda t}\; \text{QFT}^\dagger \ket{s}.
\end{equation}
The above operation $e^{-i\Lambda t}$ requires controlled quantum logic gates, which depend on the Fourier coefficients; this requires encoding of the Fourier coefficients in the computational basis, as performed by the QFTC algorithm.

In the following, we will demonstrate how the QFTC algorithm can be used to simulate Hamiltonians with a circulant matrix structure, as shown in \autoref{fig:diagonal}, with the aid of the quantum circuit given in simulating diagonal Hamiltonians \cite{childs2004}:

\begin{figure}[ht]
\[
    \Qcircuit @C=1em @R=1em {
     \lstick{\ket{s}} 	& \gate{\text{QFT}^\dagger} 	& \ctrl{1} 	& \qw 	& \ctrl{1} 	&\gate{\text{QFT}}	&\rstick{e^{-iCt}\ket{s}} \qw \\				
     \lstick{\ket{0}} 	&\qw 	&\multigate{4}{\text{QFTC}}		& \gate{e^{+2^{L/2}it \ket{1}\bra{1}}}		& \multigate{4}{\text{QFTC}^\dagger}	&\qw	& \rstick{\ket{0}} \qw\\		
     \lstick{\ket{0}} 	&\qw 	&\ghost{\text{QFTC}}				& \gate{e^{-2^{L/2-1}it \ket{1}\bra{1}}} 	& \ghost{\text{QFTC}^\dagger}	& \qw	&\rstick{\ket{0}} \qw\\	
     \lstick{\ket{0}} 	&\qw 	&\ghost{\text{QFTC}}				& \gate{e^{-2^{L/2-2}it \ket{1}\bra{1}}} 	& \ghost{\text{QFTC}^\dagger}	& \qw	&\rstick{\ket{0}} \qw\\	
     \lstick{\vdots~} 	& 	&\pureghost{\text{QFTC}}			& \vdots 			& \pureghost{\text{QFTC}^\dagger} 	&	& ~~~~~~\vdots\\
     \lstick{\ket{0}} 	&\qw	& \ghost{\text{QFTC}} 			& \gate{e^{-2^{L/2-p_0}it \ket{1}\bra{1}}} 		& \ghost{\text{QFTC}^\dagger}	& \qw	&\rstick{\ket{0}} \qw\\	
    }
\]
\caption{Simulation of circulant Hamiltonians. $p_0+1$ is the number of digits of the resulting Fourier coefficients and $F_k$ was encoded in the form $f_0.f_1f_2\cdots f_{p_0}$ as the complemental code for a number between $-1$ and $1$. Here we define $\text{QFTC}\ket{k}\ket{0} = \ket{k}\ket{F_k}$ (detailed in \ref{item:prep_ancilla}--\ref{item:uncompute} in \autoref{sec:qftc}).}
\label{fig:diagonal}
\end{figure}
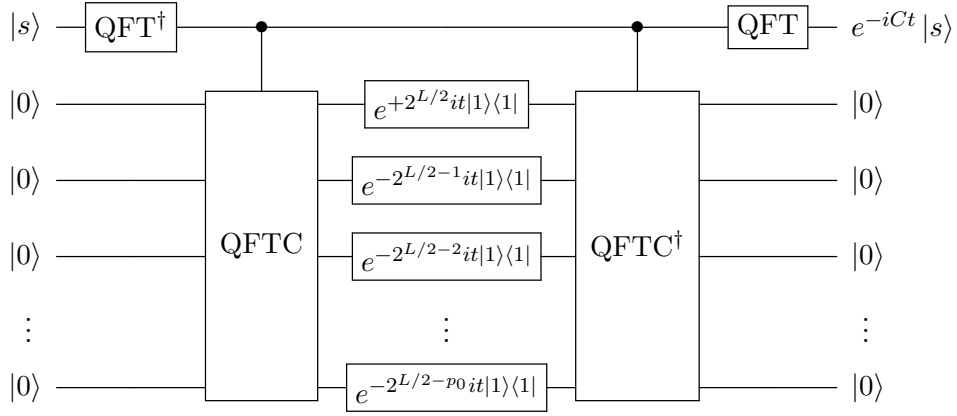

\begin{enumerate}[label=\emph{Step~\arabic*}]
\item
Perform the inverse QFT on $\ket{s}$:
\begin{equation}
\sum_{k=0}^{N-1} s_k \ket{k} \rightarrow \sum_{k=0}^{N-1} \mathfrak{s}_k \ket{k}.
\end{equation}

\item
Apply the QFTC algorithm (\ref{item:prep_ancilla}--\ref{item:uncompute} in \autoref{sec:qftc}) for $\gv{c}$:
\begin{equation}
\sum_{k=0}^{N-1} \mathfrak{s}_k \ket{k} \rightarrow \sum_{k=0}^{N-1} \mathfrak{s}_k \ket{k}\ket{F_k}.
\end{equation}

\item
Do controlled phase gate $e^{+2^{L/2}it\ket{1}\bra{1}}$ on the first digit (qubit) of $\ket{F_k}$ and $e^{-2^{L/2-p+1}it\ket{1}\bra{1}}$ on the $p$th digit (qubit) of $\ket{F_k}$ for all $p>1$:
\begin{equation}
\sum_{k=0}^{N-1} \mathfrak{s}_k \ket{k}\ket{F_k} \rightarrow \sum_{k=0}^{N-1} \mathfrak{s}_k e^{-i\Lambda_k t}\ket{k}\ket{F_k}.
\end{equation}

\item
Undo the QFTC for every $\ket{k}$:
\begin{equation}
\sum_{k=0}^{N-1} \mathfrak{s}_k e^{-i\Lambda_k t}\ket{k}\ket{F_k} \rightarrow \sum_{k=0}^{N-1} \mathfrak{s}_k e^{-i\Lambda_k t}\ket{k}.
\end{equation}

\item
Perform the QFT:
\begin{equation}
\sum_{k=0}^{N-1} \mathfrak{s}_k e^{-i\Lambda_k t}\ket{k} \rightarrow e^{-iCt}\ket{s}.
\end{equation}

\end{enumerate}

\begin{theorem}[Simulation of Circulant Hamiltonians]
\label{thm:hamiltonian}
The simulation of a circulant Hamiltonian $e^{-iCt}$
can be performed within error $\delta$
\footnote{$\|e^{-iCt} - \widetilde{e^{-iCt}}\|\leq \delta$, where $\widetilde{e^{-iCt}}$ represents the operator that is actually performed by this algorithm.}
using
$\mathcal{O}\big(\sqrt{N}t(\log N)^2/\delta^{3/2}\big)$ one- or two-qubit gates,
as well as $\mathcal{O}\big(\sqrt{N}t/\delta^{3/2})$ calls of controlled-$O_x$ and its inverse, where $\gv{x} = \gv{c}$ is a unit vector in $\mathbb{C}^{N}$ and $C$ is Hermitian.
\end{theorem}
\begin{proof}

The error present in the Hamiltonian simulation is fully determined by the precision of the QFTC algorithm. According to the above QFTC complexity analysis, we need
$\mathcal{O}\big((\log N)^2/(\delta\varepsilon)\big)$ one- or two-qubit gates, as well as $\mathcal{O}\big(1/(\delta\varepsilon))$ calls of controlled-$O_x$ and its inverse, to achieve accuracy $\varepsilon$ in $F_k$.
The fidelity achieved for the Hamiltonian simulation, as defined by the squared modulus of inner product, is
\begin{equation}
(1-\delta)^2 \abs{\langle e^{-i\tilde{C}t}\ket{s}, e^{-iCt}\ket{s} \rangle} = (1-\delta)^2 \abs{\sum_{k=0}^{N-1} e^{i (\tilde{\Lambda}_k-\Lambda_k) t} \abs{\mathfrak{s}_k}^2}
> 1 - \mathcal{O}((\sqrt{N}t\varepsilon)^2 + \delta),
\end{equation}
where the last inequality is derived using
\begin{equation}
\abs{e^{i\gamma_1}+ \abs{\Gamma} e^{i\gamma_2}} = \big( 1 + \abs{\Gamma}^2 + 2 \abs{\Gamma} \cos(\gamma_1-\gamma_2) \big)^{1/2}
> (1 + \abs{\Gamma}) \abs{\cos \frac{\gamma_1 - \gamma_2}{2}},
\end{equation}
and $\tilde{\Lambda}_k$ are the estimated (truncated) eigenvalues calculated via the QFTC algorithm. For a fixed $\delta$ in the QFTC algorithm, if we choose $\varepsilon = \sqrt{\delta}/(\sqrt{N}t)$,
the fidelity will be $1 - \mathcal{O}(\delta)$. We then need $\mathcal{O}\big((\log N)^2/(\delta\varepsilon)\big)$ =
$\mathcal{O}\big(\sqrt{N}t(\log N)^2/\delta^{3/2}\big)$ one- or two-qubit gates,
as well as $\mathcal{O}\big(\sqrt{N}t/\delta^{3/2})$ calls of controlled-$O_x$ and its inverse.

\end{proof}

\section{Conclusion}
\label{sec:conclusion}

In this paper, we proposed a new QFTC algorithm, an efficient quantum algorithm to encode the results of the Discrete Fourier Transform
in the computational basis. This algorithm allows us to overcome a main shortcoming of the conventional Quantum Fourier Transform --
the inability to perform operations controlled by the Fourier coefficients.
In short, the QFTC utilizes swap tests to obtain a function of
the Fourier coefficients in the amplitudes, with individual coefficients then extracted via amplitude estimation and quantum arithmetic.

Secondly, a detailed complexity analysis of the QFTC algorithm was performed, finding it require
$\mathcal{O}\big((\log N)^2/(\delta\epsilon)\big)$ calls of one- or two-qubit gates, as well as  $\mathcal{O}\big(1/(\delta\epsilon)\big)$ calls of controlled-$O_x$ and its inverse, in order to achieve fidelity $1-\delta$ and  precision $\epsilon$.
Note that the overall complexity depends polylogarithmically on $N$, similarly to the conventional QFT, and we require only controlled phase gates and Hadamard gates.
The inverse proportionality with the desired accuracy, $\epsilon$,
occurs due to the application of amplitude estimation within the algorithm.

Finally, we detailed an application of the QFTC algorithm in the simulation of circulant Hamiltonians, which requires
$\mathcal{O}\big(\sqrt{N}t(\log N)^2/\delta^{3/2}\big)$ one- or two-qubit gates,
as well as $\mathcal{O}\big(\sqrt{N}t/\delta^{3/2})$ calls of controlled-$O_x$ and its inverse to achieve fidelity $1- \delta$.
This paves the way for a quantum circuit implementation of continuous-time quantum walks on circulant graphs,
with potential applications in a wide array of disciplines. Further applications of the QFTC algorithm are also expected.

\section{Acknowledgements}
The authors would like to thank Ashley Montanaro for constructive comments, and Jeremy O'Brien, Jonathan Matthews, Xiaogang Qiang, Lyle Noakes for helpful discussions.

\bigskip

\appendix

\section{Quantum Arithmetic}
\label{sec:arithmetic}

Addition and multiplication are basic elements of arithmetic in classical computer.
There have been several proposals on how to build quantum adders and multipliers \cite{cuccaro2004,draper2006,alvarez-sanchez2008,vedral1996},
constructed predominately using CNOT gates and Toffoli gates.
Draper's addition quantum circuits, however, utilizes the quantum Fourier transformation (QFT) \cite{draper2000}.
QFT-based multiplication and related quantum arithmetic have also been proposed \cite{ruiz-perez2014,maynard2013,maynard2013-a,pavlidis2014}.
In this appendix, for completeness, we outline the construction of the quantum arithmetic gates required for the QFTC algorithm in detail.

We show here, using QFT-based circuits and fixed-point number representation,
that all elementary quantum arithmetic gates used to construct the QFTC circuit (including adders, multipliers and cosine gates) have $\mathcal{O}(\mathrm{poly}(n))$ complexity,
when $n$ is the number of qubits (number of digits) representing a number.
With accuracy $\epsilon$, this results in $\mathcal{O}(\mathrm{polylog}(1/\epsilon))$ complexity.

\subsection{QFT Multiply-adder}
\label{subsec:QMA}

We begin by describing a quantum multiply-adder for real inputs $a$ and $b$ between 0 and 1.
Let $\ket{a} = \ket{a_1}\ket{a_2}\cdots\ket{a_m}$ represent the fixed point number $a = 0.a_1 a_2 \cdots a_m$ (same for $b$).
Using this representation, the quantum multiply-adder (QMA), as shown in \autoref{fig:QMA1}(a), can realize the following transformation,
\begin{equation}
\label{eq:multiply-adder}
\Pi^\pm_{m,n}\ket{a}\ket{b}\ket{c} = \ket{a}\ket{b}\ket{c \pm a\cdot b},
\end{equation}
where $m$ and $n$ denote the number of digits of $a$ and $b$ respectively.

In quantum multiply-adders, the outputs, unlike the inputs, can be negative and we use the complemental code $c^{(C)} = c_0.c_1 c_2 \cdots c_{m+n} \in [0,2)$ to represent the output $c \in (-1,1)$ and $c = c^{(C)}$ if $c$ is non-negative and $c = c^{(C)}-2$ if $c$ is negative. $\ket{c}$ is composed of $\ket{c_0}\ket{c_1}\cdots\ket{c_{m+n}}$. Note that this quantum multiply-adder also applies to any fixed-point-represented numbers by cleverly choosing the appropriate positions of the fractional points.


\begin{figure}[htbp]
\centering
\subfigure[~quantum multiply-adder]{
\includegraphics[width=4.9in]{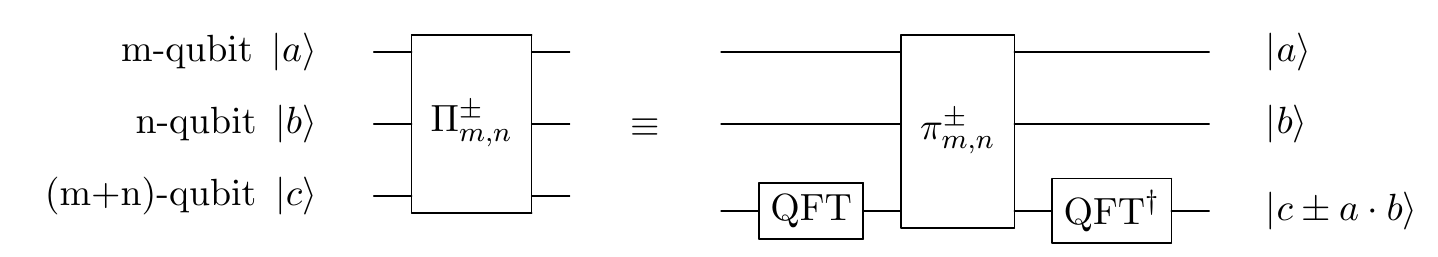}}\\
\subfigure[~intermediate multiply-adder]{
\includegraphics[width=2.2in]{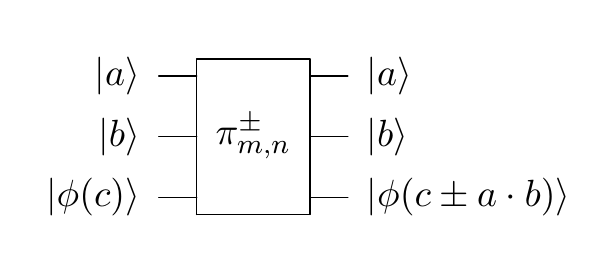}}
\caption{Quantum circuit of the multiply-adder}
\label{fig:QMA1}
\end{figure}

The quantum multiply-adder can be decomposed into the following form, as shown in \autoref{fig:QMA1}(b):
\begin{equation}
	\Pi^\pm_{m,n} = (\mathbb{I}\otimes \mathbb{I}\otimes \mathrm{QFT}^\dagger)\cdot \mathfrak{\pi}^\pm_{m,n} \cdot (\mathbb{I}\otimes \mathbb{I}\otimes \mathrm{QFT}),
\end{equation}
where $\mathfrak{\pi}^\pm_{m,n}$ represents an intermediate quantum multiply-adder,
\begin{equation}
	\pi^\pm_{m,n}\ket{a}\ket{b}\ket{\phi(c)}=\ket{a}\ket{b}\ket{\phi(c\pm a\cdot b)}
\end{equation}
with $\ket{\phi(c)}:=\mathrm{QFT}\ket{c}$ and $\ket{\phi_{k}(c)} = \frac{1}{\sqrt{2}}(\ket{0}+e^{2\pi i c\cdot2^{m+n-k}}\ket{1})$, $k=1,2,\cdots,m+n+1$.

\autoref{fig:QMA2} shows a detailed quantum circuit construction of $\mathfrak{\pi}^\pm_{m,n}$, using the QFT adders $2^{-l} \Sigma^\pm_{m,n}$, which act as follows:
\begin{equation}
\label{eq:A}
\boxed{2^{-l} \Sigma^\pm_{m,n}}\ket{b}\ket{\phi(c)} = \ket{b}\ket{\phi(c\pm2^{-l} b)}.
\end{equation}
The QFT adders are constructed via controlled phase operations, as shown in \autoref{fig:subfig:Acircuit}.

\newsavebox{\smlmat}
\savebox{\smlmat}{$R^\pm_k =\left(\begin{smallmatrix}1 & 0 \\0 & \exp(\pm2\pi i /2^k)\end{smallmatrix}\right)$}

\begin{figure}[htbp]
\centering
\subfigure[~ $\mathfrak{\pi}^\pm_{m,n}$ gate]{
\label{fig:subfig:QMAcircuit} 
\includegraphics[width=6in]{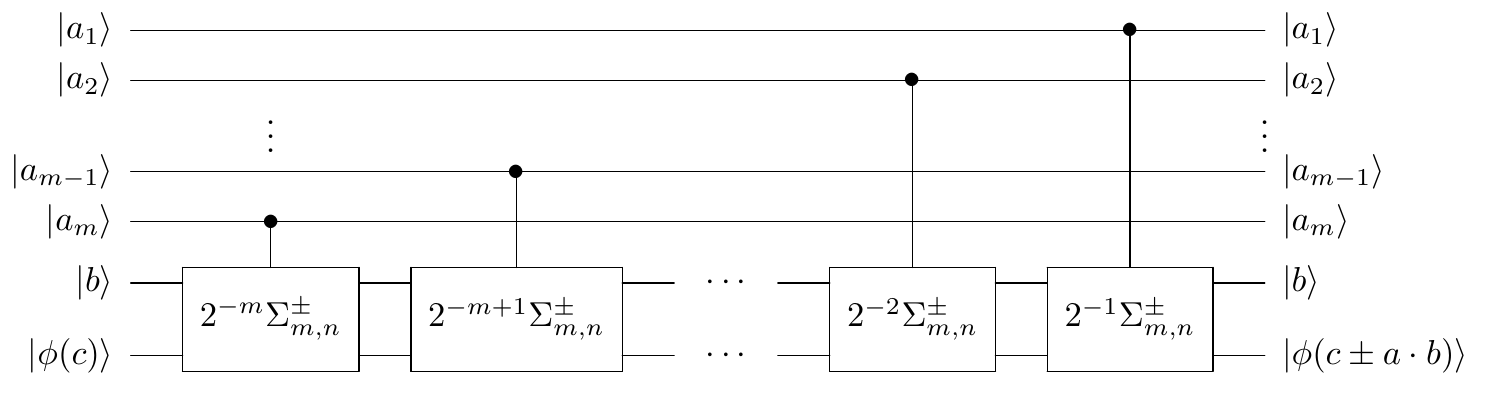}}
\subfigure[~ QFT adder]{
\label{fig:subfig:A}  
\includegraphics[width=2.5in]{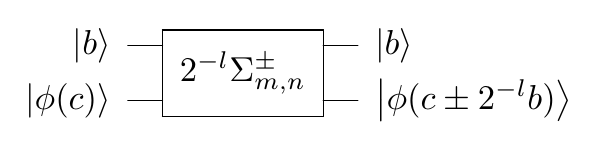}}
\subfigure[~Detailed quantum circuit construction of the QFT adder $2^{-l} \Sigma^\pm_{m,n}$, $R^\pm_k = \ket{0}\bra{0}+e^{\pm2\pi i/2^{k}}\ket{1}\bra{1}$.]{
\label{fig:subfig:Acircuit}  
\includegraphics[width=6.4in]{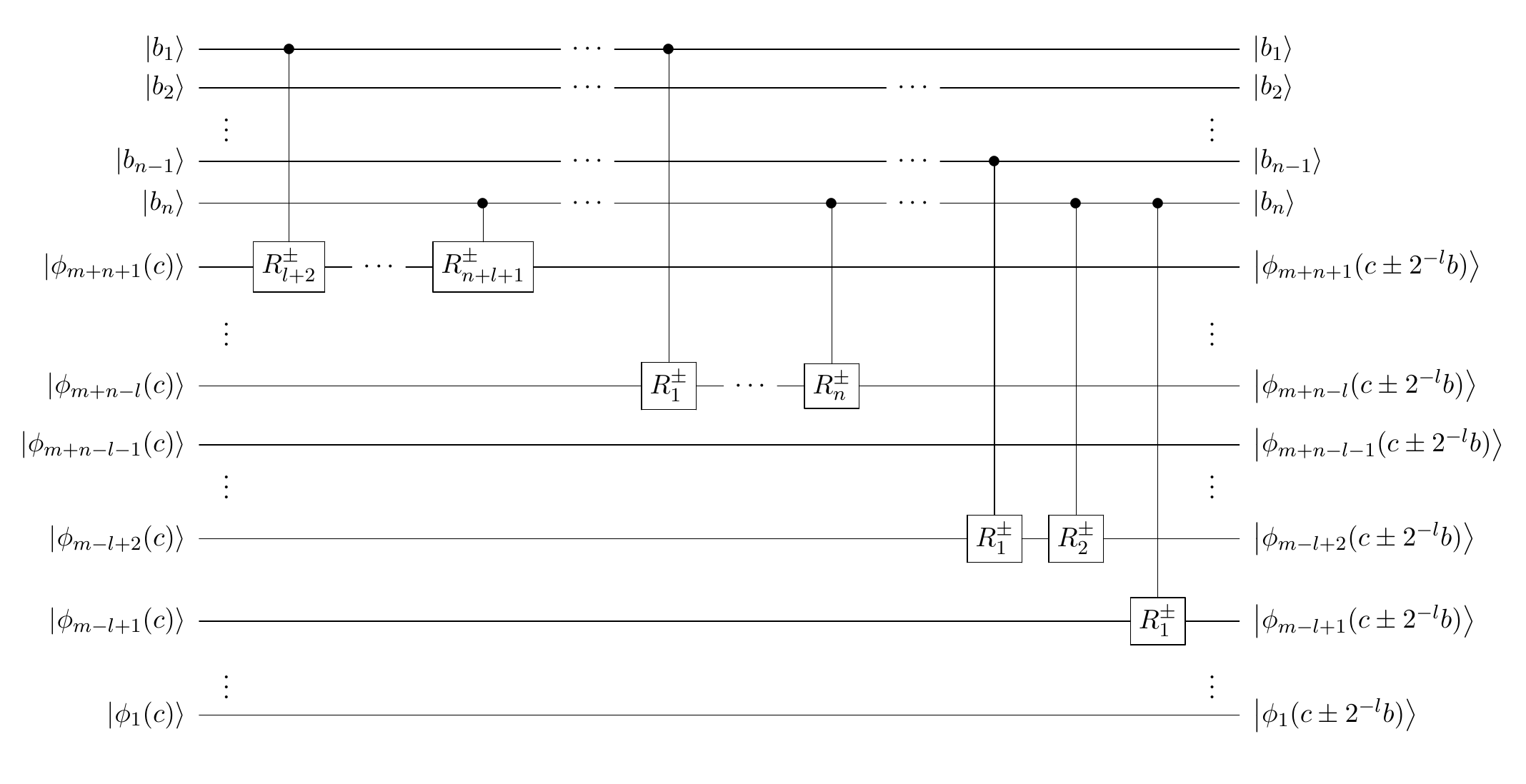}}
\caption{Quantum circuit of $\mathfrak{\pi}^\pm_{m,n}$} \label{fig:QMA2} 
\end{figure}

After applying the QFT adder $2^{-m} \Sigma^\pm_{m,n}$ (controlled by $\ket{a_m}$) in \autoref{fig:subfig:QMAcircuit}, we get
\begin{equation}
\ket{\phi(c)} \longrightarrow \ket{\phi(c \pm a_m 2^{-m} b)}.
\end{equation}
Proceeding in a similar fashion, it can be seen that the final output state of the intermediate multiply-adder is
\begin{equation}
\ket{\phi(c + a_m 2^{-m} b +\cdots+ a_1 2^{-1}b)} = \ket{\phi(c \pm a \cdot b)}.
\end{equation}

To illustrate how the circuit works, take for example the evolution of $\phi_{m+n-l}(c)$ after $R^\pm_1,\ldots,R^\pm_n$:
\begin{equation}
\ket{0} + e^{2\pi i c \cdot 2^{l}}\ket{1}
\longrightarrow
\ket{0} + e^{2\pi i c \cdot 2^{l} \pm b }\ket{1}.
\end{equation}
We then have $$\ket{\phi_{k}(c)} \rightarrow \ket{\phi_{k}(c\pm 2^{-l}b)}.$$

It is clear from \autoref{fig:subfig:Acircuit} that the QFT adder uses $\mathcal{O}\big((m+n)n\big)$ one- or two-qubit gates.
Hence, the total complexity of the intermediate QFT multiply-adder's is $\mathcal{O}\big((m+n)mn\big)$.
Thus, with QFT scaling $\mathcal{O}\big((m+n)^2\big)$, the total complexity of the quantum multiply-adder $\Pi^\pm_{m,n}$ is $\max\{\mathcal{O}(mn^2),\mathcal{O}(nm^2)\}$.

Note that if we choose $l=0$ in $2^{-l} \Sigma^\pm_{m,n}$ and perform a QFT and an inverse QFT before and after the application of the QFT adder in \autoref{eq:A}, we have a quantum adder
\begin{equation}
\label{eq:adder}
\ket{b}\ket{c} \rightarrow \ket{b}\ket{c\pm b}.
\end{equation}
We can also add (or subtract) two numbers without having to destroy their original values encoded in the computational basis, i.e.,
\begin{equation}
\ket{b}\ket{c}\ket{0} \rightarrow \ket{b}\ket{c}\ket{b} \rightarrow \ket{b}\ket{c}\ket{b \pm c}
\end{equation}
by using \autoref{eq:adder} twice.

\subsection{Quantum sine \& cosine gate}
\label{subsec:cosgate}

By implementing the Taylor series using the quantum multiply-adder, we are able to build a quantum sine (and cosine) gate.
Suppose $x = 0. x_1 x_2 \cdots x_n$ and $x \in [0,1)$. We aim to build a sine gate calculating the value of $\sin \pi x$, performing $\ket{x}\ket{0^n}\ket{0^m}\rightarrow \ket{x}\ket{\sin \pi x}\ket{\Psi^{\text{ancilla}}}$.

\begin{figure}[htbp]
\centering
\subfigure[~ sine gate]{
\label{fig:subfig:Scircuit}  
\includegraphics[width=6.4in]{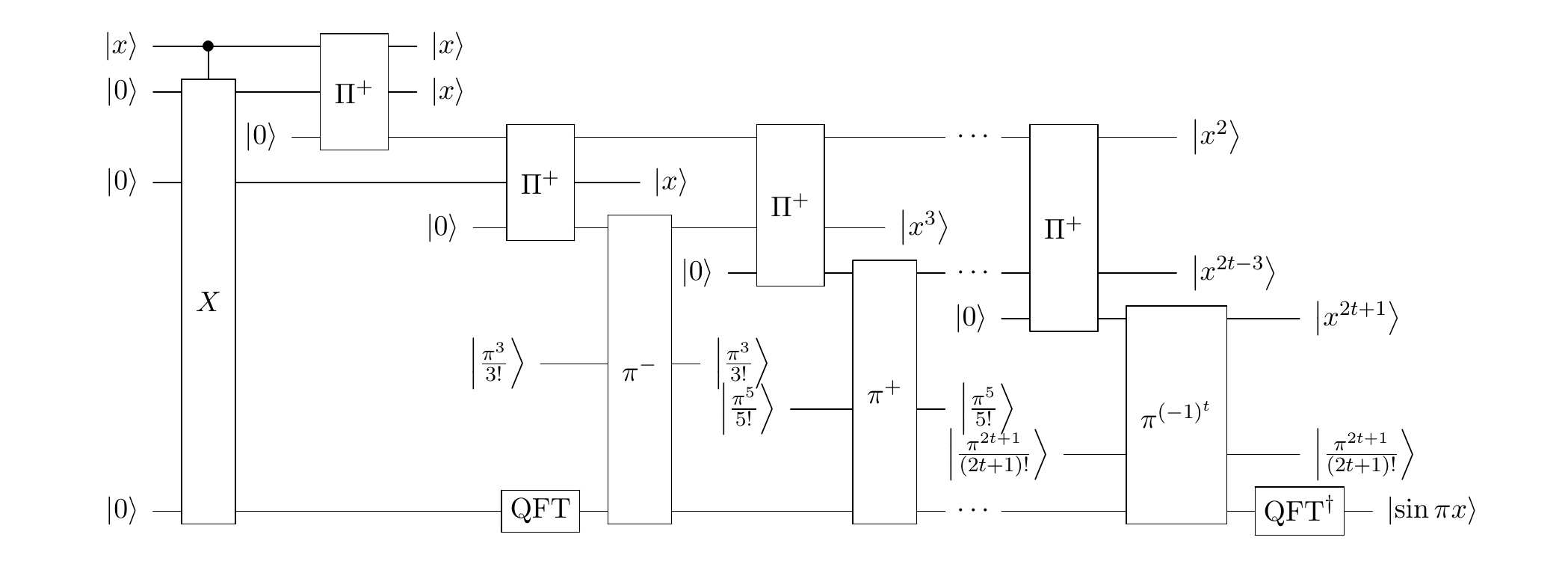}}
\subfigure[~ cosine gate]{
\label{fig:subfig:Ccircuit} 
\includegraphics[width=6.4in]{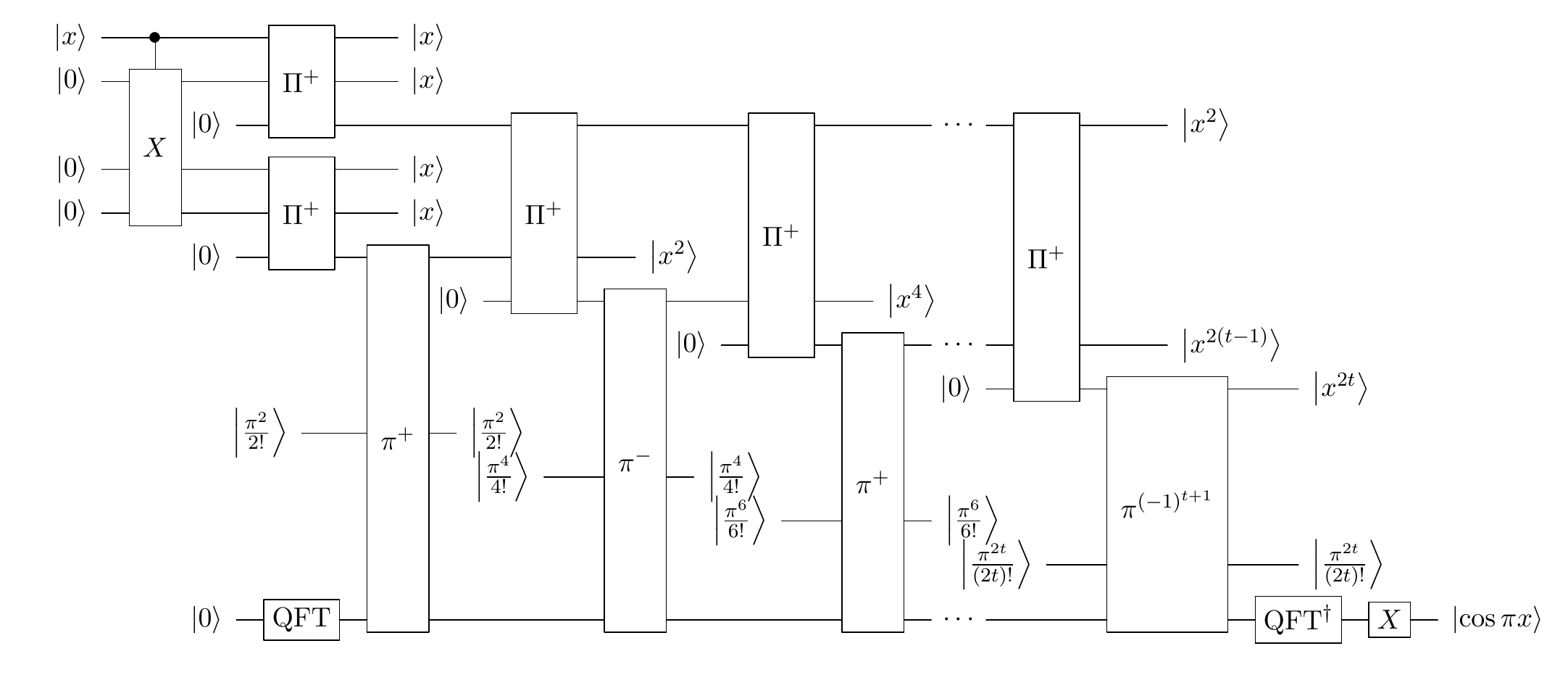}}
\caption{Quantum circuits of the sine and cosine gates ($\ket{0}$ represents a number of qubits in above circuits where the numbers are omitted). Pauli-X gates are used to transform $\ket{0}$ into $\ket{x}$ and the subscript for all the quantum multiply-adders in above circuits is $(p',p')$.} \label{fig:CScircuit} 
\end{figure}

We now consider the error in the truncated Taylor series. First, the error introduced by imprecision in the $n$-digit representation of $x$ is $\mathcal{O}(2^{-n})$, since the derivative of $\sin \pi x$ is bounded. The Taylor series of $\sin \pi x$ at around $x=0$ is
\begin{equation}
\label{eq:SIN}
\sin \pi x = \pi x - \frac{(\pi x)^3}{3!} + \frac{(\pi x)^5}{5!} - \cdots + (-1)^{t}\frac{(\pi x)^{2t+1}}{(2t+1)!} + \frac{(-1)^{t+1}\cos \pi z}{(2t+3)!}(\pi x)^{(2t+3)}.
\end{equation}
The remainder term for the $k$th term in the expansion is $\frac{f^{(k+1)}(z)}{(k+1)!}x^{k+1}$, where $z\in (0, x)$, according to Taylor's Theorem \cite{kline1998calculus}.
As a result, in Eq.~\eqref{eq:SIN}, the reminder term (error) is $\frac{(-1)^{t+1}\cos \pi z}{(2t+3)!}(\pi x)^{(2t+3)}$ and is obviously bounded by $\mathcal{O}(2^{-n})$ for $t=n$.

In the sine gate, the $t+1$ terms $\left\{\pi x,\frac{(\pi x)^3}{3!},\cdots,(-1)^{t}\frac{(\pi x)^{2t+1}}{(2t+1)!}\right\}$ are first calculated and then added (or subtracted) together. Suppose each of the $t+1$ terms has an error within $2^{-p}$. Taking $p = n + \lceil\log n\rceil = \mathcal{O}(n)$, the error introduced by adding and subtracting will be $\mathcal{O}(t\cdot 2^{-p}) = \mathcal{O}(2^{-n})$. Suppose all multiply-adders have $p'$ digits inputs.
When errors in $y_1,y_2$ are within $2^{-(\ell+1)}$ and $y_1,y_2\leq 1-2^{-(\ell+1)}$,
$(y_1+2^{-(\ell+1)})(y_2+2^{-(\ell+1)})
= y_1y_2 + 2^{-\ell}(y_1+y_2)/2  + 2^{-2\ell-2} \leq y_1y_2 + 2^{-\ell}.$
It means that by applying the multiply-adders $2t$ times, the error will be $2^{2t}$ times larger. Thus we can choose a $p'=\mathcal{O}(p+2t)=\mathcal{O}(n)$ which guarantees accuracy $2^{-p}$ in all the powers of $x$ and also all the $t+1$ terms in the Taylor series.

We conclude that we can choose $t = \mathcal{O}(n)$ and $p' = \mathcal{O}(n)$ so that the total accuracy of the sine gate is bounded by $2^{-n}$.
\autoref{fig:CScircuit} shows the quantum circuit for the sine and cosine gate.
The complexity of the quantum sine gate can be calculated based on the scaling of quantum multiply-adders which equals to $\mathcal{O}(p'^3)$.
The total complexity of the quantum sine gate is $\mathcal{O}(t p'^3) = \mathcal{O}(n^4)$ for accuracy $2^{-n}$. To put it in another way, $\mathcal{O}(\text{polylog}(1/\epsilon))$ one- or two-qubit gates are required to achieve accuracy $\epsilon$.

\section{Implementing circulant operators}
\label{sec:circulantoperator}

\begin{figure}[ht]
\[
\Qcircuit @C=1em @R=1em{
\lstick{\ket{s}}	&	{/}\qw	&	\gate{\text{QFT}^\dagger}	&	\ctrl{1}	&	\gate{\text{QFT}}	&	\rstick{C\ket{s}}\qw	\\
\lstick{\ket{0}}	&	{/}\qw	&	\gate{O_x}	&	\gate{R}	&	\gate{H^{\otimes L}}	&	\meter\qw^{\qquad\qquad\qquad\quad\ket{0^L}}	\\
}\]
\caption{Implementation of circulant matrices. Here $R\ket{k}\ket{j} = e^{2\pi ikj/N}\ket{k}\ket{j}$.}
\label{fig:cmatrix}
\end{figure}
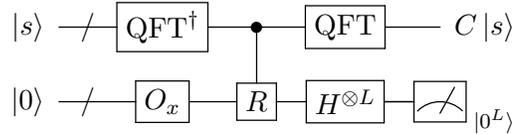

Consider an arbitrary state $\ket{s}$. We wish to obtain $C\ket{s}$, where $C$ is an arbitrary circulant matrix.
Below, we present a possible algorithm for implementing a circulant matrix quantum operator (see \autoref{fig:cmatrix}).

\begin{enumerate}[label=\emph{Step~\arabic*}]
\item
Perform the inverse QFT on $\ket{s}$:
\begin{equation}
\sum_{k=0}^{N-1} s_k \ket{k} \rightarrow \sum_{k=0}^{N-1} \mathfrak{s}_k \ket{k}.
\end{equation}

\item
Add another register prepared to $\sum_{j=0}^{N-1} c_j\ket{j}$ using $O_x$ ($\gv{x}=\gv{c}$ in \autoref{eq:oracle}):
\begin{equation}
\sum_{k=0}^{N-1} \mathfrak{s}_k \ket{k} \rightarrow \sum_{j,k=0}^{N-1} \mathfrak{s}_k c_j\ket{k}\ket{j}.
\end{equation}

\item
Apply the controlled phase gate so that $\ket{k}\ket{j} \rightarrow e^{2\pi ikj/N}\ket{k}\ket{j}$:
\begin{equation}
\sum_{j,k=0}^{N-1} \mathfrak{s}_k c_j\ket{k}\ket{j} \rightarrow \sum_{j,k=0}^{N-1} \mathfrak{s}_k c_j e^{2\pi ijk/N}\ket{k}\ket{j}.
\end{equation}

\item
Apply Hadamard gates to $\ket{j}$:
\begin{equation}
\sum_{j,k=0}^{N-1} \mathfrak{s}_k c_j e^{2\pi ijk/N}\ket{k}\ket{j} \rightarrow \sum_{j,k=0}^{N-1} \mathfrak{s}_k \ket{k}(F_k \ket{0^L}+ \sqrt{1-F_k^2}\ket{0^\perp}),
\end{equation}
where $\ket{0^\perp}$ represents any states perpendicular to $\ket{0^L}$.

\item
\label{item:post-select}
By post-selecting the ancillary qubit state $\ket{0^L}$, the quantum state in the first register collapses to
\begin{equation}
\frac{1}{\sqrt{\sum_{k}|F_k \mathfrak{s}_k|^2}}
\sum_{k=0}^{N-1} F_k \mathfrak{s}_k \ket{k}.
\end{equation}

\item
Perform the QFT:
\begin{equation}
\text{QFT}\sum_{k=0}^{N-1} \mathfrak{s}_k F_k \ket{k} \propto C\ket{s}.
\end{equation}

\end{enumerate}

Note that the post-selection probability of obtaining the correct state in \ref{item:post-select} is
\begin{equation}
p = \sum_{k=0}^{N-1} \abs{\mathfrak{s}_k F_k}^2,
\end{equation}
and $p$ equals to $1/N$ when $C$ is unitary.
Therefore, using amplitude amplification \cite{brassard2002},
$\mathcal{O}((\log N)^2/\sqrt{p})$ one- or two-qubit gates, as well as
$\mathcal{O}(1/\sqrt{p})$ calls of $O_x$, $O_s$ and their inverses are needed to implement a circulant matrix operation $C$, where $O_s \ket{0^L} = \sum_{k=0}^{N-1}s_k\ket{k}$.

\bibliography{refs}

\end{document}